\documentclass[letterpaper,11pt]{article}
\usepackage{amsmath}
\usepackage{amssymb}
\usepackage{amsfonts}
\usepackage{setspace}
\usepackage{bbm}
\usepackage{color}
\usepackage{enumerate}
\usepackage{mathrsfs}
\usepackage[longnamesfirst]{natbib}
\usepackage{graphicx}
\usepackage{hyperref}
\usepackage{enumitem}
\usepackage[margin=1.00in]{geometry}
\newtheorem{theorem}{Theorem}[section]
\newtheorem{lemma}{Lemma}[section]

\newtheorem{assumption}{Assumption}[section]
\newtheorem{corollary}{Corollary}[section]
\newtheorem{definition}{Definition}
\newtheorem{remark}{Remark}[section]
\newenvironment{proof}[1][Proof]{\noindent \textbf{#1.} }{\  \rule{0.5em}{0.5em}}
\setlist[enumerate]{topsep=-2pt,itemsep=0pt,parsep=0pt,partopsep=0pt,label=(\roman*), align=left, leftmargin=*}

\begin{document}

\author{%
Timothy M. Christensen\thanks{%
Department of Economics, New York University, 19 W. 4th Street, 6th floor, New York, NY 10012, USA. E-mail address: \texttt{timothy.christensen@nyu.edu}
}
}

\title{%
Nonparametric identification of positive eigenfunctions\thanks{%
This paper is based on the third chapter of my doctoral dissertation at Yale. I am very grateful to my advisors Xiaohong Chen and Peter C.B. Phillips for their support, encouragement, and many valuable discussions. I would also like to thank a co-editor and two anonymous referees for helpful comments and suggestions. 
While writing this paper I was generously supported by a Carl Arvid Anderson Fellowship from the Cowles Foundation.
}
}

\date{First version October 2012; Revised December 2013, August 2014}

\maketitle

\begin{abstract} 
\noindent Important features of certain economic models may be revealed by studying positive eigenfunctions of appropriately chosen linear operators. Examples include long-run risk-return relationships in dynamic asset pricing models and components of marginal utility in external habit formation models. This paper provides identification conditions for positive eigenfunctions in nonparametric models. Identification is achieved if the operator satisfies two mild positivity conditions and a power compactness condition. Both existence and identification are achieved under a further non-degeneracy condition. The general results are applied to obtain new identification conditions for external habit formation models and for positive eigenfunctions of pricing operators in dynamic asset pricing models. 

\bigskip 


\bigskip

\end{abstract}

\pagenumbering{arabic}

\newpage
\section{Introduction}

Recent work in economics has shown that important features of certain economic models may be may be revealed by studying positive eigenfunctions of appropriately chosen linear operators. For example, \cite{HansenScheinkman2009} use a positive eigenfunction of pricing operators to decompose stochastic discount factors into permanent and transitory components. This decomposition presents a convenient device for studying long-run risk-return relationships (see also \cite{AlvarezJermann,Hansen2012,BackusChernovZin}). Moreover, \cite{HansenScheinkman2012} provide existence and identification conditions for the value function in models with recursive preferences by studying a related positive eigenfunction. \cite{Chenetal2012} show that an external habit formation function in a semiparametric consumption capital asset pricing model (C-CAPM) may be expressed as a positive eigenfunction of an operator related to the Euler equation. Further, \cite{Ross2011} recovers a probability distribution over future states implicit in option prices using a positive eigenfunction and its eigenvalue. 

The purpose of this paper is to provide a set of tractable conditions for nonparametric identification of positive eigenfunctions in economic models.\footnote{The conditions identify the positive eigenfunction up to scale (any positive multiple of a positive eigenfunction is a positive eigenfunction). If the positive eigenfunction is normalized to have unit norm, then the conditions are sufficient for the point identification of the normalized positive eigenfunction.} Previously, \cite{Chenetal2012} and \cite{Christensen2012} have provided identification conditions in the context of a semiparametric C-CAPM and a discrete-time implementation of the long-run analysis of \cite{HansenScheinkman2009}, respectively. Independently, \cite{LintonLewbelSrisuma2011} and \cite{EscancianoHoderlein} have studied identification of marginal utility of consumption in nonparametric Euler equations.
Each of these papers derives identification conditions by first restricting the parameter space to consist of positive functions belonging to a $L^p$ space and then applying results from integral operator theory. In contrast, \cite{HansenScheinkman2009} do not restrict the parameter space ex ante and instead apply Markov process theory to derive identification conditions for positive eigenfunctions of pricing operators in continuous-time models.  

Identification of positive eigenfunctions has been long studied in abstract functional analysis in the mathematics literature (see, e.g., \cite{KreinRutman,Schaefer1960}). The function-analytic identification conditions usually relate to properties of the resolvent of the operator. As such, these conditions are generally difficult to motivate and interpret in an economic context.

This paper aims to help bridge the gap between the high-level conditions in functional analysis and the primitive conditions for integral operators previously studied by providing a reasonably tractable set of identification conditions for the positive eigenfunction of a positive operator on a Banach lattice. When applied to the special case of operators on $L^2$ spaces, the identification conditions herein are weaker than those previously provided by \cite{Chenetal2012} and \cite{Christensen2012}.\footnote{It should be noted, however, that the identification conditions in \cite{Christensen2012} are a simple set of conditions that are also used to study estimation and long-run approximation. Similarly, \cite{Chenetal2012} sought to provide sufficient conditions that were similar to conditions imposed in the literature on estimation.} In particular, the identification conditions in this paper require a weaker form of compactness and are not cast in terms of restrictions on conditional probability densities. This generalization is important because compactness may fail and/or certain conditional densities may not exist in some applications. The general results presented in this paper are applied to derive new identification conditions for two models for which the previous identification conditions may be violated.

Additional properties of the positive eigenfunction and its eigenvalue are established under the identification conditions. These properties are useful for nonparametric estimation of the positive eigenfunction and its eigenvalue. In particular, if the identification conditions hold then the eigenvalue corresponding to the positive eigenfunction is the largest eigenvalue of the operator, and it is an isolated, simple eigenvalue. This final property may be used to establish continuity of the eigenvalue and eigenfunction with respect to perturbations of the operator (see, e.g., Chapter IV of \cite{Kato}). If an estimator of the operator is near the true operator in an appropriate sense then the maximum eigenvalue of the estimator, and its eigenfunction, should (under regularity) also be near the true eigenvalue/eigenfunction. \cite{Christensen2012} uses these properties to derive large-sample theory for nonparametric estimators of the positive eigenfunction and its eigenvalue.

The general conditions are applied to obtain new identification conditions for positive eigenfunctions of pricing operators in dynamic asset pricing models. Identification conditions are presented for stationary discrete- and continuous-time environments. Identification is shown to hold under a no-arbitrage condition, a weak dependence condition on the state process, and a power compactness condition. Existence and identification is established under an additional condition on the yield on zero-coupon bonds. These identification conditions complement those that \cite{HansenScheinkman2009} provide for possibly nonstationary, continuous-time environments. Identification conditions for a C-CAPM with external habit formation are also presented as an application.

In the applications dealt with in this paper, the eigenfunction equation can be interpreted as a conditional moment restriction. Identification of conditional moment restriction models has received much attention of late, particularly in the context of nonparametric models with endogeneity (see, e.g., \cite{SeveriniTripathi2006,dHaultfoeuille2011,Chenetal2012} and references therein). A conditional moment restriction model is identified when there is a unique solution to the conditional moment equation. In contrast, in this paper there are a continuum of solutions to the conditional moment equation when the positive eigenfunction is identified.\footnote{This is because an eigenfunction is only identified up to scale.} Further, here the identification results involve a crucial shape restriction, namely positivity, which is not typically imposed in the literature on identification of conditional moment restriction models.

The paper is organized as follows. Section \ref{specsec} reviews relevant concepts from spectral theory. Section \ref{sL2} provides identification conditions for positive operators on $L^p$ spaces. Section \ref{slrr} applies the conditions to pricing operators, and Section \ref{habit} to external habit formation models. Section \ref{sgen} presents identification conditions for operators on Banach lattices. An appendix contains additional background material on Banach lattices and all proofs.

\section{Review of some relevant definitions}\label{specsec}

Let $E$ be a Banach space and $T : E \to E$ be a bounded linear operator.\footnote{The definitions are from \cite{Schaefer1999} and assume $E$ is a Banach space over $\mathbb C$. When $E$ is a Banach space over $\mathbb R$, the complex extension $T(x+iy) = T(x) + iT(y)$ for $x,y \in E$ of $T$ is defined on the complexification $E + i E$ of $E$. The spectrum and associated quantities of $T$ when $E$ is a Banach space over $\mathbb R$ are obtained by applying the complex Banach space definitions to the complex extension of $T$.} The resolvent set $\rho(T) \subseteq \mathbb C$ of $T$ is the set of all $z \in \mathbb C$ for which the {resolvent operator} $R(T,z) := (T - zI)^{-1}$ exists as a bounded linear operator on $E$ (where $I:E \to E$ is the identity operator, i.e., $Ix = x$ for all $x \in E$). The spectrum $\sigma(T)$ is defined as the complement of $\rho(T)$ in $\mathbb C$, i.e. $\sigma(T) := (\mathbb C \setminus \rho(T))$. The point spectrum $\pi(T)\subseteq \sigma(T)$ of $T$ is the set of all $z \in \mathbb C$ for which the nullspace of $(T - zI)$ is nontrivial. Each $\lambda \in \pi(T)$ is an eigenvalue of $T$ and any nonzero $\psi$ in the nullspace of $(T - \lambda I)$ is an eigenvector of $T$ corresponding to $\lambda$. The geometric multiplicity of $\lambda \in \pi(T)$ is the dimension of the nullspace of $(T - \lambda I)$. The algebraic multiplicity of $\lambda \in \sigma(T)$ is the order of the pole of $R(T,z)$ at $z = \lambda$. If $\lambda \in \pi(T)$ has algebraic and geometric multiplicity equal to 1 then $\lambda$ is said to be \emph{simple}. Further, $\lambda \in \pi(T)$ is said to be \emph{isolated} if $\inf_{z \in \sigma(T) : z \neq \lambda} |z - \lambda| > 0$.
 The spectral radius $r(T)$ of $T$ is defined as $r(T) := \sup\{|\lambda| : \lambda \in \sigma(T)\}$.

\section{Identification in $L^p$ spaces}\label{sL2}

This section provides identification conditions for positive eigenfunctions when the parameter space consists of positive functions in a $L^p$ space. The first set of identification conditions presented are for $L^p$ spaces with $1 \leq p < \infty$. Identification conditions are also presented for the space of bounded functions, which may be of interest in certain applications. 

Let $(\mathcal X,\mathscr X,\mu)$ be a $\sigma$-finite measure space. To simplify notation, in what follows let $L^p$ denote the space $L^p(\mathcal X,\mathscr X,\mu)$.\footnote{The space $L^p(\mathcal X,\mathscr X,\mu)$ with $1 \leq p < \infty$ consists of all (equivalence classes of) measurable functions $f : \mathcal X \to \mathbb R$ such that $\int |f|^p \,\mathrm d\mu < \infty$. The space $L^\infty(\mathcal X,\mathscr X,\mu)$ consists of all (equivalence classes of) measurable functions $f: \mathcal X \to \mathbb R$ such that $\mathrm{ess} \sup| f| < \infty$.} Let $T : L^p \to L^p$ be a linear operator. $T$ is said to be \emph{positive} if $T f \geq 0$ a.e.-$[\mu]$ whenever $f \geq 0$ a.e.-$[\mu]$. 

The cases $1 \leq p < \infty$ and $p = \infty$ are dealt with separately because of the different properties of the $L^p$ ($1 \leq p < \infty$) and $L^\infty$ spaces.

\subsection{Case 1: $1 \leq p < \infty$}

The following conditions are sufficient for both existence and nonparametric identification of the positive eigenfunction of $T$ when the parameter space consists of positive functions in a $L^p$ space with $1 \leq p < \infty$. 

\begin{assumption}\label{aL2}
$T : L^p \to L^p$ is a bounded linear operator such that:\\[-20pt]
\begin{enumerate}
\item $T$ is positive,
\item for each $f \in L^p$ with $f \geq 0$ a.e.-$[\mu]$ and $f \neq 0$ there exists $n \geq 1$ such that $T^n f > 0$ a.e.-$[\mu]$, 
\item $T^n$ is compact for some $n \geq 1$, and 
\item $r(T) > 0$. 
\end{enumerate}
\end{assumption}

These conditions will be referred to as \emph{positivity}, \emph{eventual strong positivity}, \emph{power compactness}, and \emph{non-degeneracy}. The two {positivity} conditions are typically straightforward to motivate, either from the economic context of the problem or the structure of the operator. The power compactness condition is trivially satisfied if $T$ is compact, though this more general condition suffices. 
The non-degeneracy condition is satisfied if $T$ has a nonzero eigenvalue. Non-degeneracy is also satisfied if there exists a function $f \in L^p$ with $f \geq 0$ a.e.-$[\mu]$ and $f \neq 0$ such that $T^n f \geq \delta f$ a.e.-$[\mu]$ for some $n \geq 1$ and $\delta > 0$ \citep[Proposition 3]{Schaefer1960}. Both the non-degeneracy and eventual strong positivity conditions  have an economic interpretation in the asset pricing application below.

The following Theorem also shows that Assumption \ref{aL2} is sufficient for identification of the positive eigenfunctions of both $T$ and its adjoint $T^*$. To introduce the adjoint, let $q = \infty$ if $p = 1$, otherwise let $q$ be such that $q^{-1} + p^{-1} = 1$. The space $L^q$ can be identified as the dual space of $L^p$. The adjoint $T^*: L^q \to L^q$ of $T$ is a bounded linear operator defined by the equality
\begin{equation} \label{adjoint}
 \int g (Tf) \,\mathrm d\mu = \int (T^*g) f \,\mathrm d\mu
\end{equation}
for all $f \in L^p$ and $g \in L^q$.

\begin{theorem}\label{tL2}
Let Assumption \ref{aL2} hold. Then:\\[-20pt]
\begin{enumerate}
\item there exist $\bar f \in L^p$ with $\bar f > 0$ a.e.-$[\mu]$ and $\bar f^* \in L^q$ with $\bar f^* > 0$ a.e.-$[\mu]$ such that $T \bar f = r(T) \bar f$ and $T^* \bar f^* = r(T) \bar f^*$
\item $\bar f$ and $\bar f^*$ are the unique eigenfunctions of $T$ and $T^*$ in $L^p$ and $L^q$, respectively, that are non-negative a.e.-$[\mu]$
\item $r(T)$ is simple and isolated.
\end{enumerate}
\end{theorem}

Theorem \ref{tL2} is both an identification and existence result. Assumptions \ref{aL2}(i)(ii)(iii) are enough to establish identification, as made clear in the following Corollary.

\begin{corollary}\label{cL2}
Let Assumptions \ref{aL2}(i)(ii)(iii) hold and let $\bar f$ be a positive eigenfunction of $T$. Then Assumption \ref{aL2}(iv) is satisfied and the conclusions of Theorem \ref{tL2} hold.
\end{corollary}

\subsection{Case 2: $p = \infty$}

Identification conditions are now presented for the case in which the parameter space is the space $L^\infty$ of (essentially) bounded functions. There are two important differences from the $1 \leq p < \infty$ case. First, the eventual strong positivity condition needs to be strengthened. Second, identification of the adjoint eigenfunction is not considered because the dual of $L^\infty$ is typically identified with a space of signed measures. 

\begin{assumption}\label{aLinfty}
$T : L^\infty(\mu) \to L^\infty(\mu)$ is a bounded linear operator such that:\\[-20pt]
\begin{enumerate}
\item $T$ is positive,
\item for each $f \in L^\infty$ with $f \geq 0$ a.e.-$[\mu]$ and $f \neq 0$ there exists $n \geq 1$ such that $\mathrm{ess}\inf T^n f > 0$, 
\item $T^n$ is compact for some $n \geq 1$, and 
\item $r(T) > 0$. 
\end{enumerate}
\end{assumption}

As with the preceding case, positivity and non-degeneracy may be motivated by the economic context of the problem or the structure of the operator. 

\begin{theorem}\label{tLinfty}
Let Assumption \ref{aLinfty} hold. Then:\\[-20pt]
\begin{enumerate}
\item there exists $\bar f \in L^\infty$ with $\mathrm{ess}\inf \bar f > 0$ such that $T \bar f = r(T) \bar f$
\item $\bar f$ is the unique eigenfunction of $T$ in $L^\infty$ that is non-negative a.e.-$[\mu]$
\item $r(T)$ is simple and isolated.
\end{enumerate}
\end{theorem}

Theorem \ref{tLinfty} establishes both existence and identification of the positive eigenfunction. As with the $1 \leq p < \infty$ case, Assumption \ref{aLinfty}(i)(ii)(iii) are sufficient to establish identification.

\begin{corollary}\label{cLinfty}
Let Assumptions \ref{aLinfty}(i)(ii)(iii) hold and let $\bar f$ be a positive eigenfunction of $T$. Then Assumption \ref{aLinfty}(iv) is satisfied and the conclusions of Theorem \ref{tLinfty} hold.
\end{corollary}

\subsection{Discussion of closely related results}\label{comparison}

\cite{Chenetal2012} and \cite{Christensen2012} provide identification conditions for $L^2$ spaces in which the linear operator is of the form 
\begin{equation}
 T f(x) = \int k(x,y) f(y)\,\mathrm d \mu(y)\,.
\end{equation}
For identification, they assume that the integral kernel $k : \mathcal X \times \mathcal X \to \mathbb R$ satisfies:\\[-20pt]
\emph{\begin{enumerate}
\item[(a)] $k(x,y) > 0$ a.e.-$[\mu \otimes \mu]$, and 
\item[(b)] $\int \int k(x,y)^2 \,\mathrm d\mu(x) \,\mathrm d \mu(y) < \infty$.
\end{enumerate}}
Condition (a) guarantees positivity and eventual strong positivity (Assumption \ref{aL2}(i)(ii)). Condition (b) implies $T$ is compact and therefore power compact (Assumption \ref{aL2}(iii)).\footnote{\cite{Christensen2012} also establishes identification and existence of positive eigenfunctions of integral operators on $L^p$ spaces under a weaker power compactness condition than (b).} Conditions (a) and (b) also imply $r(T) > 0$ (Assumption \ref{aL2}(iv); see Theorem V.6.6 of \cite{Schaefer1974}). Therefore, conditions (a) and (b) are sufficient for both existence and identification of the positive eigenfunction in the space $L^2$.

Identification using conditions (a) and (b) may be problematic for certain models including asset pricing models in higher-order Markov environments (Section \ref{AR(p)}) and external habit formation models where the habit component is a function of consumption growth over two or more periods (Section \ref{habit}). First, certain conditional densities may fail to exist, making conditions (a) and (b) less interpretable as restrictions on probability densities. Second, the operator may not be compact, thereby violating condition (b) (see Lemma \ref{noncompact}). Theorem \ref{tL2} is applied to provide new identification conditions in these applications (see Sections \ref{AR(p)} and \ref{habit}).

\section{Application: dynamic asset pricing models and the long run}\label{slrr}

\cite{HansenScheinkman2009} (HS hereafter) present a framework for decomposing stochastic discount factors (SDFs) in dynamic asset pricing models into permanent and transitory components (see also \cite{AlvarezJermann,Hansen2012,BackusChernovZin}).  This decomposition presents a convenient device for extracting information about long-run risk-return relationships, even in complicated nonlinear environments where extrapolation to the long run might otherwise be difficult. The key tools for performing their analysis are a positive eigenfunction and eigenvalue of a collection of pricing operators.

HS derive identification conditions for the positive eigenfunction in general continuous-time environments using stochastic stability results for continuous-time Markov processes. They show that the positive eigenfunction compatible with their stochastic stability conditions must be unique, but they do not rule out the existence of other positive eigenfunctions. 

Theorem \ref{tL2} is now applied to derive new identification conditions for the positive eigenfunction in discrete- and continuous-time environments. Some of the identification conditions are stronger than those in HS. For instance, here a power compactness condition is imposed and the environment is assumed to be stationary, though stationarity is not necessarily restrictive in an asset pricing context.\footnote{Stationarity may sometimes be achieved by a change of variables. For example, consumption-based asset pricing models are often written in terms of consumption growth to avoid potential nonstationarity in aggregate consumption \citep{HansenSingleton,GallantTauchen}).} Under stationarity, the parameter space may be restricted ex ante to consist of positive functions belonging to an appropriately chosen $L^p$ space. This restriction, together with the identification conditions presented below, both guarantees uniqueness of the positive eigenfunction and shows how its eigenvalue is related to the spectrum of the pricing operators. Conditions for existence of the positive eigenfunction in discrete-time environments are also presented. 

\subsection{Identification in discrete-time environments}

The discrete-time environment is first described, following \cite{HansenScheinkman2012,HansenScheinkman2013}. Let $\{(X_t,Y_t)\}_{t=-\infty}^\infty$ denote a strictly stationary, discrete-time, first-order Markov process with support $\mathcal X \times \mathcal Y \subseteq \mathbb R^{d_x} \times \mathbb R^{d_y}$. Assume further that the joint distribution of $(X_{t+1},Y_{t+1})$ conditioned on $(X_t,Y_t)$ depends only on $X_t$. The vector $X_t$ will summarize all the relevant information for pricing at date $t$, and will be referred to as the state vector. 

To describe the pricing operators, let $\psi : \mathcal X \to \mathbb R$ be a measurable function. Assume that the date-$t$ price of a claim to $\psi(X_{t+1})$ payable at time $t+1$ is given by the Euler or no-arbitrage equation
\begin{equation} \label{euler}
 T \psi(x) = \mathbb E\left[m(X_t,X_{t+1},Y_{t+1}) \psi(X_{t+1})|X_t = x\right]
\end{equation}
where $m : \mathcal X \times \mathcal X \times \mathcal Y \to \mathbb R$ is a measurable function, the random variable $m(X_t,X_{t+1},Y_{t+1})$ is the SDF for pricing single-period claims at date $t$ and $m(X_t,X_{t+1},Y_{t+1}) \geq 0$ a.e.-$[\overline Q]$ where $\overline Q$ is the unconditional distribution of $(X_t,X_{t+1},Y_{t+1})$.\footnote{The joint process $(X_t,Y_t)$ is used following \cite{HansenScheinkman2012,HansenScheinkman2013}. This general framework clearly nests environments in which the state process is $\{X_t\}$ and the SDF is $m(X_t,X_{t+1})$.}  Equation (\ref{euler}) shows that the pricing operator $T$ is a linear operator on the space of payoff functions $\psi$. The $n$ period pricing operator $T_n$ may be similarly defined for $n \geq 1$ as the operator that assigns date-$t$ prices to claims to $\psi(X_{t+n})$ payable at time $t+n$. For instance, the date $t$ price of a zero-coupon bond which matures in $n$ periods is given by $T_n \psi(x)$ when $\psi(x) = 1$ for all $x \in \mathcal X$. Assuming trading at intermediate dates and time homogeneity of $m$, it follows that $T_n$ takes the form
\begin{equation} \label{euler2}
 T_n \psi(x) = \mathbb E\left[\left. \left( \prod_{s = 0}^{n-1} m(X_{t+s},X_{t+s+1},Y_{t+s+1}) \right) \psi(X_{t+n})\right|X_t = x\right]\,.
\end{equation}
The structure placed on the SDF and the state process implies that each $T_n$ factorizes as $T_n = T^n$, i.e., $T_n \psi$ is obtained by iteratively applying $T$ to $\psi$ for $n$ times.

An example of this setup is the representative-agent C-CAPM, in which 
\begin{equation} \label{C-CAPM}
 m(X_t,X_{t+1},Y_{t+1}) = \beta ( C_{t+1}/C_t )^{-\gamma}
\end{equation}
where $C_t$ is aggregate consumption at time $t$ and $C_{t+1}/C_t = g(X_t,X_{t+1},Y_{t+1})$ for some measurable function $g$. This specification offers a good deal of flexibility for modeling consumption growth, including $X_t = (G_t,Z_t)$ where $Z_t$ is a vector of variables such that $(G_t,Z_t')'$ characterizes all the relevant information for the evolution of the state from $t$ to $t+1$. For example, if $G_t$ evolves as a $\ell$th order Markov process then $Z_t = (G_{t-1},\ldots,G_{t-\ell+1})'$.

As the state process is assumed stationary, the following identification conditions are presented for the space $L^p(\mathcal X,\mathscr X,Q)$ for some $1 \leq p < \infty$, where $\mathscr X$ is the Borel $\sigma$-algebra on $\mathcal X$ and $Q$ is the stationary distribution of $X$. To simplify notation, let $L^p$ denote the space $L^p(\mathcal X,\mathscr X,Q)$.

\begin{definition} \label{principal eigenfunction definition}
$\phi \in L^p$ is a \emph{principal eigenfunction} of the collection of pricing operators $\{T_n : n \geq 1\}$ with eigenvalue $\rho$ if $\phi > 0$ a.e.-$[Q]$ and $T_n \phi = \rho^n \phi$ for each $n \geq 1$.
\end{definition}

\begin{remark} \label{equiv remark}
Let $T \phi = \rho \phi$. The factorization $T_n = T^n$ implies that $T_n \phi = \rho^n \phi$. Therefore, $\phi$ is a principal eigenfunction of $\{T_n : n \geq 1\}$ if and only if $\phi$ is a positive eigenfunction of $T$.
\end{remark}

HS show that $\rho$ and $\phi$ characterize the pricing of long-horizon assets. For example, under stochastic stability conditions, HS obtain the limiting result
\begin{equation} \label{limiting}
 \lim_{n \to \infty} \rho^{-n} T_n \psi(x) = \widetilde{ \mathbb E}[\psi(X)/\phi(X)] \phi(x)
\end{equation}
where $\widetilde {\mathbb E}[\cdot]$ is an expectation under a change of measure induced by the permanent component of the SDF. Therefore, $\rho$ encodes the yield on long-term claims to state-contingent payoffs and $\phi$ captures state-dependence of prices of such claims. \cite{Christensen2012} shows that, in stationary environments, the expectation $\widetilde{\mathbb E}[ \cdot]$ may be expressed in terms of $Q$, $\phi$, and the positive eigenfunction $\phi^*$ of the adjoint of $T$. The pair $\rho$ and $\phi$ may also be used to decompose the SDF into the product of its permanent and transitory components, and to place bounds on the size of the permanent and transitory components (see \cite{AlvarezJermann}; HS; \cite{Hansen2012,BackusChernovZin}). Uniqueness of the positive eigenfunction is clearly essential to the analysis of HS and others.

Identification conditions for discrete-time environments are now derived using Theorem \ref{tL2}.

\begin{assumption}\label{lrr cpt} $\{T_n : n \geq 1\}$ satisfies the following:\\[-20pt]
\begin{enumerate}
\item $T : L^p \to L^p$ is a bounded linear operator
\item $T_n$ is compact for some $n \geq 1$.
\end{enumerate}
\end{assumption}

It remains to verify the remaining conditions of Theorem \ref{tL2}. Clearly $T : L^p \to L^p$ is positive, as it is the composition of two positive operators, namely the multiplication operator $M_m$ given by $M_m \psi = m \psi$ where $m \geq 0$ a.e.-$[\overline Q]$ and the conditional expectation operator $\mathbb E [\cdot|X_t = x]$.

The following weak dependence and no-arbitrage conditions are sufficient for eventual strong positivity.

\begin{assumption}\label{irred}
For any $S \in \mathscr X$ with $Q(S) > 0$ there exists an integer $n = n(S) \geq 1$ such that $\Pr(X_{t+n} \in S | X_t = x) > 0$
a.e.-$[Q]$.
\end{assumption}

Assumption \ref{irred} is similar to the irreducibility condition in HS, though they require irreducibility under a change of measure induced by the permanent component. If both the stationary density of $X_t$ and the conditional density of $X_{t+1}$ given $X_t$ exist and are strictly positive then Assumption \ref{irred} is trivially satisfied with $n = 1$. However, the conditional density will fail to exist if the state process $X_t$ is formed by stacking a higher-order Markov process as a first-order Markov process, in which case the extra generality of Assumption \ref{irred} is useful (see Section \ref{AR(p)}).

The principle of no-arbitrage asserts that any claim to a non-negative payoff which is positive with positive conditional probability must command a positive price (see, e.g., \cite{HansenRenault} and references therein). It is helpful to think of the collection 
\begin{equation}
 \Psi = \{ \psi \in L^p :  \psi \geq 0\mbox{ a.e.-}[Q], \psi \neq 0\}
\end{equation}
as a menu of claims to future state-contingent payoffs. That is, at date $t$ each pair $(\psi,n) \in \Psi \times \mathbb N$ represents a claim to $\psi(X_{t+n})$ payable at time $t+n$. The following no-arbitrage type condition is imposed on (\ref{euler2}).

\begin{assumption}\label{noarb}
For each $(\psi,n) \in \Psi \times \mathbb N$, $\mathbb E[\psi(X_{t+n}) |X_t = x] > 0$ a.e.-$[Q]$ implies $T_n \psi(x) > 0$ a.e.-$[Q]$.
\end{assumption}

\begin{lemma}\label{suff conds}
If Assumptions \ref{irred} and \ref{noarb} hold then $T$ satisfies eventual strong positivity.
\end{lemma}

HS assume that their (continuous-time) SDF process is positive almost surely. The following remark shows that a discrete-time version of their condition implies Assumption \ref{noarb}.

\begin{remark}\label{esp}
Let $\overline Q_n$ denote the joint distribution of $(X_t,X_{t+1},Y_{t+1},\ldots,X_{t+n},Y_{t+n})$. Suppose that $\prod_{i=0}^{n-1} m(X_{t+i},X_{t+1+i},Y_{t+1+i}) > 0$ a.e.-$[\overline Q_n]$ for each $n \geq 1$. Then  Assumption \ref{noarb} is satisfied.
\end{remark}

An identification result is presented first, for which some extra notation is required. Let $q = \infty$ if $p = 1$, otherwise let $q$ be such that $q^{-1} + p^{-1} = 1$, and let $L^q$ denote the space $L^q(\mathcal X, \mathscr X, Q)$. Let $T^* : L^q \to L^q$ denote the adjoint of $T$ defined in (\ref{adjoint}), and note that $T_n^* = (T^*)^n$ as a consequence of the factorization $T_n = T^n$.

\begin{theorem}\label{id na disc}
Let Assumptions \ref{lrr cpt}, \ref{irred} and \ref{noarb} hold and let $\phi \in L^p$ be a principal eigenfunction of $\{T_n : n \geq 1\}$ with eigenvalue $\rho$. Then:\\[-20pt]
\begin{enumerate}
\item $\phi$ is the unique eigenfunction of $\{T_n : n \geq 1\}$ in $L^p$ that is non-negative a.e.-$[Q]$
\item $\rho$ is positive, simple, isolated, and $\rho = r(T)$
\item there exists $\phi^* \in L^q$ with $\phi^* > 0$ a.e.-$[Q]$ such that $T_n^* \phi^* = \rho^n \phi^*$ for each $n \geq 1$
\item $\phi^*$ is the unique eigenfunction of $T^*$ in $L^q$ that is non-negative a.e.-$[Q]$.
\end{enumerate}
\end{theorem}

Theorem \ref{id na disc} shows that if one can be guaranteed of the existence of a principal eigenfunction of the collection of operators $\{T_n : n \geq 1\}$ then this principal eigenfunction must be unique. Moreover, it shows that the eigenvalue $\rho$ is the largest eigenvalue of $T$. 

As in Section \ref{sL2}, an extra non-degeneracy condition is ensures existence and uniqueness of the principal eigenfunction. Here the non-degeneracy condition may be formulated in terms of the yield on zero-coupon bonds. Let $1(\cdot) : \mathcal X \to \mathbb R$ denote the constant function, i.e. $1(x) = 1$ for all $x \in \mathcal X$. The price of a zero-coupon bond maturing in $n$ periods is $T_n 1(x)$, and its yield-to-maturity $y_n(x)$ is defined by the relation
\begin{equation}
 T_n 1(x) = (1+y_n(x))^{-n}\,.
\end{equation}
or, equivalently, $y_n(x) = [T_n 1(x)]^{-1/n}-1$.

\begin{assumption}\label{nondeg}
There exists $n \geq 1$ and $C < \infty$ such that $y_n(x) \leq C$ a.e.-$[Q]$.
\end{assumption}

\begin{lemma}\label{nondeg proof}
$r(T) > 0$ under Assumption \ref{nondeg}.
\end{lemma}

Assumption \ref{nondeg} is satisfied with $n = 1$ for the C-CAPM in (\ref{C-CAPM}) if expected consumption growth $\mathbb E[g(X_t,X_{t+1},Y_{t+1})|X_t=x]$ is uniformly bounded away from infinity a.e.-$[Q]$. Assumption \ref{nondeg} is also trivially satisfied in models with a constant risk free rate. Assumption \ref{nondeg} can be weakened to require only that there exists $\psi \in \Psi$ and $n \in \mathbb N$ such that $T_n \psi(x) \geq \delta \psi(x)$ a.e.-$[Q]$ for some $\delta > 0$.

\begin{theorem}\label{ex id na disc}
Let Assumptions \ref{lrr cpt}, \ref{irred}, \ref{noarb}, and \ref{nondeg} hold. Then:\\[-20pt]
\begin{enumerate}
\item there exists a principal eigenfunction $\phi \in L^p$ of $\{T_n : n \geq 1\}$
\item the conclusions of Theorem \ref{id na disc} hold.
\end{enumerate}
\end{theorem}

Similar results to Theorems \ref{id na disc} and \ref{ex id na disc} are presented in \cite{Christensen2012}. However, there the positivity, eventual strong positivity, and non-degeneracy conditions are satisfied by imposing primitive restrictions on the kernel representing the operator $T$.

\subsection{Identification with higher-order Markov processes}\label{AR(p)}

Higher-order Markov processes are popular modeling devices as they maintain the tractability of Markov processes while allowing for richer dynamics than first-order processes. Examples include higher-order vector autoregressive processes (VARs), higher-order nonlinear VARs (\cite{HTY1998} and references therein), autoregressive Gamma processes \citep{GourierouxJasiak2006}, and compound autoregressive processes \citep{DarollesGourierouxJasiak2006}. See, e.g., \cite{MonfortPegoraro2007}, \cite{BertholonMonfortPegoraro2008} and \cite{Eraker2008} for applications of these processes to term structure models.

To fix ideas, in what follows let $\{W_t\}_{t=-\infty}^\infty$ be a strictly stationary Markov process of finite order $\ell>1$ where each $W_t$ has full support $\mathbb R$ (the following analysis extends directly to the multivariate case, except the notation is more complicated). Setting $X_t = (W_t,W_{t-1},\ldots,W_{t-\ell+1})'$ makes $\{X_t\}_{t=-\infty}^\infty$ a strictly stationary  first-order Markov process, where each $X_t$ has support $\mathbb R^\ell$. Let $Q$ denote the unconditional distribution of $X_t$ on $\mathbb R^\ell$. 

The identification conditions presented in Section \ref{comparison} run into difficulty in such environments for two reasons. First, the conditional density of $X_{t+1}$ given $X_t$ does not exist: the elements $(W_t,\ldots,W_{t-\ell+2})$ of $X_{t+1}$ are known when $X_{t}$ is known. Second, and more critically, the operator $T$ may fail to be compact, thereby violating the integrability condition (b) in Section \ref{comparison}. The following Lemma describes one situation in which compactness of $T$ may fail.

\begin{lemma}\label{noncompact}
Let $X_t = (W_t,W_{t-1},\ldots,W_{t-\ell+1})'$ and let $\{W_t\}_{t=-\infty}^\infty$ be as described above. If there exists $c > 0$ such that $\mathbb E[m(X_t,X_{t+1},Y_{t+1})|X_t = x] \geq c$ a.e.-$[Q]$  then $T : L^2 \to L^2$ is not compact.
\end{lemma}

Theorem \ref{id na disc} is now applied to derive identification conditions in higher-order environments. Of course, Assumptions \ref{lrr cpt}, \ref{irred}, and \ref{noarb} may be assumed directly in which case identification is immediate by Theorem \ref{id na disc}. To illustrate the power of the new identification results, further primitive conditions on the process $\{W_t\}_{t=-\infty}^\infty$ and the SDF $m$ are now derived. These are but one example of sufficient conditions for Assumptions \ref{lrr cpt} and \ref{irred} in higher-order Markov environments and are in no way exhaustive.

\begin{assumption}\label{aWlin} $\{W_t\}_{t=-\infty}^\infty$ is a strictly stationary nonlinear AR$(\ell)$ process of the form
\begin{equation}
 W_{t+1} = h( W_t , \ldots , W_{t-\ell+1} ) + u_{t+1}
\end{equation}
where $h : \mathbb R^{\ell} \to \mathbb R$ is a measurable function and the innovations $u_{t}$ are i.i.d. with $\mathbb E[u_t] = 0$ and density $f_U$ such that $f_U(u) > 0$ for all $u \in \mathbb R$. 
\end{assumption}

\begin{lemma} \label{lem-irred}
Let Assumption \ref{aWlin} hold. Then Assumption \ref{irred} is satisfied.
\end{lemma}

The state process in Assumption \ref{aWlin} is similar to that recently studied in \cite{BackusChernovZin}, except here $\{W_t\}_{t=-\infty}^\infty$ evolves as a nonlinear AR$(\ell)$ process rather than a linear MA$(\infty)$ process. See \cite{FanYao} and references therein for further regularity conditions on $h$ which ensure stationarity and ergodicity of $\{W_t\}_{t=-\infty}^\infty$.

Although the conditional density of $X_{t+l}$ given $X_t$ does not exist for $l = 1,\ldots,\ell-1$, the conditional density of $X_{t+\ell}$ given $X_t$ may exist under Assumption \ref{aWlin}. For example, when $\ell = 2$
\begin{equation}
 \left( \begin{array}{c}
 W_{t+2} \\
 W_{t+1}
 \end{array} \right) =
 \left( \begin{array}{c}
 h(h(W_t,W_{t-1})+u_{t+1},W_t)+u_{t+2} \\
 h(W_t,W_{t-1})+u_{t+1}
 \end{array} \right)
\end{equation}
thus $X_{t+2} = (W_{t+2},W_{t+1})'$ is a measurable function of $X_t=(W_t,W_{t-1})'$ and the innovations $u_{t+1}$ and $u_{t+2}$. Analogous results clearly hold for $X_{t+\ell}$ when $\ell > 2$. 

The following conditions on the SDF and $\{X_t\}_{t=-\infty}^\infty$ are sufficient for $T$ to be bounded and power compact. To simplify notation, in what follows it is assumed that $m(X_t,X_{t+1},Y_{t+1}) = m(X_t,W_{t+1})$.

\begin{assumption}\label{aWsdf}
There exists a finite constant $C$ such that $\mathbb E[m(X_t,W_{t+1})^2|X_t=x] \leq C$ a.e.-$[Q]$.
\end{assumption}

\begin{assumption}\label{aWcpt}
$Q$ has density $q$, the conditional density $f(X_{t+\ell}|X_t)$  of $X_{t+\ell}$ given $X_t$ exists, and
\begin{eqnarray*}
 \int \int \left( \left(m(X_t,W_{t+1}) \cdots m(X_{t+\ell-1},W_{t+\ell}) \right) \frac{f(X_{t+\ell}|X_t)}{q(X_{t+\ell})} \right)^2 \,\mathrm dQ(X_t)\,\mathrm dQ(X_{t+\ell}) < \infty \,.
\end{eqnarray*}
\end{assumption}

Assumption \ref{aWsdf} implies that $T : L^2 \to L^2$ is bounded. Assumption \ref{aWcpt} implies that $T_\ell : L^2 \to L^2$ is compact, and is analogous to the integrability condition in \cite{Chenetal2012} and \cite{Christensen2012} but applied to the operator $T_\ell$ rather than $T$. The following identification result is a direct consequence of Theorem \ref{id na disc}.

\begin{theorem}\label{tW}
Let Assumptions \ref{noarb}, \ref{aWlin}, \ref{aWsdf} and \ref{aWcpt} hold and let $\phi \in L^2$ be a positive eigenfunction of $T$ with eigenvalue $\rho$. Then: \\[-20pt]
\begin{enumerate}
\item $\phi$ is the unique eigenfunction of $\{T_n : n \geq 1\}$ in $L^2$ that is non-negative a.e.-$[Q]$
\item $\rho$ is positive, simple, isolated, and $\rho = r(T)$.
\end{enumerate}
\end{theorem}

\subsection{Identification in continuous-time environments}

Identification conditions for continuous-time models to are now derived using Theorem \ref{tL2}. To describe the environment, let $\{Z_t\}_{t\in \mathbb R}$ be a continuous-time Markov process with support $\mathcal Z \subset \mathbb R^{d_z}$ and whose sample paths are right continuous with left limits. Following HS, consider a class of model in which the date-$t$ price of a claim to $\psi(Z_{t+\tau})$ payable at $t +\tau$ is given by
\begin{equation} \label{euler cont}
 T_\tau \psi(Z_t) = \mathbb E \left[M_\tau(Z_t) \psi(Z_{t+\tau})|Z_t\right]
\end{equation}
for each $\tau \in [0,\infty)$, where the SDF process $M_\tau(Z_t)$ is a multiplicative functional of the sample path $\{Z_s : t \leq s \leq s+\tau\}$, i.e., $M_0(Z_t) = 1$ and $M_{\tau+\upsilon}(Z_t) = M_\tau(Z_{t+\upsilon})M_\upsilon(Z_t)$ for each $\tau,\upsilon \in [0,\infty)$ (see HS for examples and further details).
 This restriction on the SDF means that the collection of pricing operators $\{T_\tau : \tau \in [0,\infty)\}$ factorizes as $T_{\tau+\upsilon} = T_\tau T_\upsilon$ for $\tau,\upsilon \geq 0$ and that $T_0 = I$, the identity operator. In particular, $T_{n \tau} = (T_{\tau})^n$ holds for any $\tau \in [0,\infty)$ and $n \in \mathbb N$.
 
For identification in possibly nonstationary continuous-time environments, HS assume that the SDF process is strictly positive, and that the state process when discretely-sampled is irreducible, Harris recurrent under a change of measure induced by the permanent component of the SDF, and that there exists a stationary distribution for the conditional expectations induced by the change of measure. 

Now consider the special case in which $\{Z_t\}_{t \in \mathbb R}$ is a stationary process.
Let $Q$ denote the stationary distribution of $Z$ on $\mathcal Z$, let $\mathscr Z$ denote the Borel $\sigma$-algebra on $\mathcal Z$, and let $L^p$ denote the space $L^p(\mathcal Z,\mathscr Z,Q)$.

\begin{definition} \label{principal eigenfunction definition cts}
$\phi \in L^p$ is a \emph{principal eigenfunction} of $\{T_\tau : \tau \in [0,\infty) \}$ with eigenvalue $\rho$ if $\phi > 0$ a.e.-$[Q]$ and $T_\tau \phi = \rho^\tau \phi$ for each $\tau \in [0,\infty)$. 
\end{definition}

\begin{remark} \label{equiv remark cts}
By Definition \ref{principal eigenfunction definition cts}, if $\phi$ is a principal eigenfunction of $\{T_\tau : \tau \in [0,\infty) \}$ then $\phi$ is a positive eigenfunction of $T_{\bar \tau}$ for each $\bar \tau > 0$.
\end{remark}

Alternative identification conditions are now derived by applying the results for discrete-time environments to the ``skeleton'' of operators $\{T_{n \bar \tau} : n \geq 1\}$ for some fixed $\bar \tau > 0$. Analysis of continuous-time Markov processes by means of their discretely-sampled skeleton is common practice. For instance, HS express some of their continuous-time identification conditions in terms of the skeleton of $Z_t$ under a change of measure induced by the permanent component.

\begin{assumption}\label{non-negc}
There exists $\bar \tau > 0$ such that Assumptions \ref{lrr cpt}, \ref{irred} and \ref{noarb} hold with $T_{\bar \tau}$ in place of $T$, $T_{n \bar \tau}$ in place of $T_n$, and $\{Z_{n \bar \tau}\}_{n = -\infty}^\infty$ in place of $\{X_t\}_{t=-\infty}^\infty$ .
\end{assumption}

\begin{theorem}\label{mktthmc}
Let Assumption \ref{non-negc} hold and let $\phi \in L^p$ be a principal eigenfunction of $\{T_\tau : \tau \in [0,\infty) \}$ with eigenvalue $\rho$. Then:\\[-20pt]
\begin{enumerate}
\item $\phi$ is the unique eigenfunction of $\{T_\tau : \tau \in [0,\infty) \}$ in $L^p$ that is non-negative a.e.-$[Q]$
\item $\rho^{\bar \tau}$ is a positive, simple, isolated eigenvalue of $T_{\bar \tau}$ and $\rho^{\bar \tau} = r(T_{\bar \tau})$ 
\item there exists $\phi^* \in L^q$ with $\phi^* > 0$ a.e.-$[Q]$ such that $T_{n\bar \tau}^* \phi^* = \rho^{n\bar \tau} \phi^*$ for each $n \geq 1$
\item $\phi^*$ is the unique eigenfunction of $T_{\bar \tau}$ in $L^q$ that is non-negative a.e.-$[Q]$.
\end{enumerate}
\end{theorem}

Assumptions \ref{irred} and \ref{noarb} applied to the skeleton $\{T_{n \bar \tau} : n \geq 1\}$ are analogous to HS's assumptions of irreducibility of the discretely-sampled state process (albeit HS impose irreducibility under a change of measure) and strict positivity of the SDF process. HS do not impose stationarity or any form of compactness. HS show that a unique principal eigenfunction satisfies their stochastic stability conditions, but cannot rule out the existence of multiple principal eigenfunctions. In contrast, Theorem \ref{mktthmc} shows that uniqueness of the principal eigenfunction can be achieved in stationary environments by restricting the parameter space to consist of positive functions in $L^p$. Theorem \ref{mktthmc} also shows how the eigenvalue $\rho$ is related to the spectrum of the pricing operators.

Unlike the discrete-time case, proving existence of a positive eigenfunction of $T_{\bar \tau}$ is not enough to show existence of an eigenfunction of the collection $\{T_\tau : \tau \in [0,\infty)\}$ because this collection is not characterized fully by $T_{\bar \tau}$. See HS for existence conditions for continuous-time environments.

\section{Application: external habit formation} \label{habit}

Euler equations have been used extensively for semi/nonparametric estimation of consumption-based asset pricing models following \cite{GallantTauchen}. The results presented in Section \ref{sL2} are now used to derive new identification conditions for a C-CAPM with external habit formation. 

As in \cite{ChenLudvigson}, assume that a representative agent's marginal utility of consumption at time $t$ is given by
\begin{equation} \label{muspec}
 MU_t = C_t^{-\gamma} (1-H(C_{t-1}/C_t,\ldots,C_{t-\ell}/C_t))^{-\gamma}
\end{equation}
where $C_t$ denotes consumption of the agent at time $t$ and $\ell \geq 1$. The intertemporal marginal rate of substitution at time $t$ is 
\begin{equation}
 \beta\frac{MU_{t+1}}{MU_t} = \beta \exp(-\gamma g_{t+1}) \frac{h(g_{t+1},\ldots,g_{t-\ell+2})}{h(g_t,\ldots,g_{t-\ell+1})}
\end{equation}
where $\beta$ is a time preference parameter, $h(g_t,\ldots,g_{t-L+1}) = (1-H(C_{t-1}/C_t,\ldots,C_{t-L}/C_t))^{-\gamma}$, and $g_t = \log(C_t/C_{t-1})$. When current and lagged consumption belong to the agent's date-$t$ information set $\mathcal F_t$, the Euler equation for the (gross) return $R_{i,t,t+1}$ on asset $i$ from time $t$ to time $t+1$, 
namely
\begin{equation}
 1 = \mathbb E\left[ \left.\beta \exp(-\gamma g_{t+1}) \frac{h(g_{t+1},\ldots,g_{t-\ell+2})}{h(g_t,\ldots,g_{t-\ell+1})} R_{i,t,t+1} \right|\mathcal F_t\right]\,,
\end{equation}
can be re-expressed as
\begin{equation}\label{e-habit}
 \beta^{-1} h(X_t) = T h(X_t)
\end{equation}
where $X_t = (g_t,\ldots,g_{t-\ell+1})'$ and the operator $T$ is given by 
\begin{equation}\label{t-habit}
 T \psi(x) = \mathbb E[\exp(-\gamma g_{t+1}) R_i(X_t,g_{t+1}) \psi(X_{t+1}) | X_t = x]
\end{equation}
and  $R_i(X_t,g_{t+1}) = \mathbb E[R_{i,t,t+1}|X_t,g_{t+1}]$. Note $h$ must be positive so that $MU$ is positive. Therefore, the habit function $h$ is a positive eigenfunction of the operator $T$ and $\beta^{-1}$ is its eigenvalue. 
 
\cite{Chenetal2012} use this positive eigenfunction representation of $h$ to study nonparametric identification of $h$ (given $\gamma$) when $\ell =1$. As discussed in Section \ref{comparison}, the identification conditions in \cite{Chenetal2012} do not necessarily carry over to the case $\ell > 1$.

Theorem \ref{tL2} is now used to derive identification conditions when $\ell > 1$. Of course, Assumption \ref{aL2} could be assumed directly in which case identification would be immediate. However, to illustrate the power of the new identification results some further low-level sufficient conditions are now derived. For illustrative purposes it is again assumed that $\{g_t\}_{t=-\infty}^\infty$ is a nonlinear AR$(\ell)$ process. With this specification of the consumption growth process, the operator $T$ in (\ref{t-habit}) is of the same form as the pricing operator analyzed in Theorem \ref{tW} with the substitutions $W_{t} = g_t$ and $m(X_t,W_{t+1}) = \exp(-\gamma g_{t+1})R_i(X_t,g_{t+1})$. 

To introduce the result, let $\{(g_t,R_{i,t,t+1})\}_{t=-\infty}^{\infty}$ be strictly stationary, let $Q$ denote the stationary distribution of $X_t$ where $X_t = (g_t,\ldots,g_{t-\ell+1})'$, let $ Q_\ell$ denote the joint distribution of $(X_t',X_{t+\ell}')'$, and let $L^2$ denote the space $L^2(\mathcal X,\mathscr X,Q)$ where $\mathcal X = \mathbb R^\ell$ is the support of each $X_t$ and $\mathscr X$ is the Borel $\sigma$-algebra on $\mathcal X$.

\begin{assumption}\label{a-habit}
$\{(g_t,R_{i,t,t+1})\}_{t=-\infty}^{\infty}$ is strictly stationary and satisfies the following:\\[-20pt]
\begin{enumerate}
\item $\{g_t\}_{t=-\infty}^\infty$ is a nonlinear AR$(\ell)$ process satisfying Assumption \ref{aWlin}
\item there exists a finite positive $C$ such that $\mathbb E[\{\exp(-\gamma g_{t+1})R_i(X_t,g_{t+1})\}^2|X_t = x] \leq C$ a.e.-$[Q]$
\item $Q$ has density $q$, the conditional density $f(X_{t+\ell}|X_t)$  of $X_{t+\ell}$ given $X_t$ exists, and
\begin{eqnarray*}
 \int \int \left( \left(e^{-\gamma g_{t+1}}R_i(X_t,g_{t+1}) \cdots e^{-\gamma g_{t+\ell}}R_i(X_{t+\ell-1},g_{t+\ell}) \right) \frac{f(X_{t+\ell}|X_t)}{q(X_{t+\ell})} \right)^2 \,\mathrm dQ(X_t) \,\mathrm d Q(X_{t+\ell}) < \infty 
\end{eqnarray*}
\item $e^{-\gamma g_{t+1}}R_i(X_t,g_{t+1}) \cdots e^{-\gamma g_{t+\ell}}R_i(X_{t+\ell-1},g_{t+\ell}) > 0$ a.e.-$[ Q_\ell]$.
\end{enumerate}
\end{assumption}

\begin{theorem}\label{t-habit}
Let Assumption \ref{a-habit} hold and let $h \in L^2$ be a solution to (\ref{e-habit}) with $h > 0$ a.e.-$[Q]$. Then $h$ is the unique solution in $L^2$  to (\ref{e-habit}) with $h \geq 0$ a.e.-$[Q]$ and $h \neq 0$.
\end{theorem}

\section{Identification in Banach lattices}\label{sgen}

Identification of positive eigenfunctions has been well studied in abstract functional analysis following \cite{KreinRutman}. However, the identification conditions in the mathematics literature are typically formulated in terms of properties of the resolvent of the operator, and are therefore difficult to motivate or interpret directly in an economic context. This section presents slightly stronger versions of these identification conditions which are arguably more applicable in economic contexts, and from which the identification theorems in the preceding sections are obtained. The results presented in this section are based on a version of the Kre\u{\i}n-Rutman theorem for Banach lattices due to \cite{Schaefer1960}. 

Let $E$ be a Banach lattice and $T : E \to E$ be a bounded linear operator. Examples of Banach lattices of functions include the $L^p$ spaces studied in Section \ref{sL2}, the Banach spaces $C_b(\mathcal X)$ of bounded continuous functions $f : \mathcal X \to \mathbb R$ on a (completely regular) Hausdorff space $\mathcal X$, and, if $\mathcal X$ is also compact, the space $C(\mathcal X)$ of continuous functions $f : \mathcal X \to \mathbb R$.\footnote{The terminology of a function space is used because the related identification problems in economics are typically in the context of operators on function spaces. More generally, there exist other Banach lattices that are not function spaces. The subsequent results apply equally to such spaces.}
Let $E_+$ and $E_{++}$ denote the positive cone in $E$ and its quasi-interior (see Appendix \ref{blapp}). Let $T^*$ denote the adjoint of $T$, $E^*$ denote the dual space of $E$, and $E^*_+$ and $E^*_{++}$ denote the positive cone in $E^*$ and its quasi-interior. The operator $T$ is said to be \emph{positive} if $T E_+ \subseteq E_+$ and \emph{irreducible} if $\sum_{n=1}^\infty \lambda^{-n} T^n f \in E_{++}$ for some $\lambda > r(T)$ and each $f \in (E_+ \setminus \{0\})$. 

\begin{assumption}\label{agenhl} $T : E \to E$ is a bounded linear operator such that:\\[-20pt]
\begin{enumerate}
\item $T$ is positive,
\item $r(T)$ is a pole of the resolvent of $T$, and 
\item $T$ is irreducible.
\end{enumerate}
\end{assumption}

\begin{theorem}[\cite{Schaefer1960}, Theorem 2]\label{tgenhl} Let Assumption \ref{agenhl} hold. Then:\\[-20pt]
\begin{enumerate}
\item there exist $\bar f \in E_{++}$ and $\bar f^* \in E_{++}^*$ such that $T \bar f = r(T) \bar f$ and $T^* \bar f^* = r(T) \bar f^*$
\item $r(T)$ is positive and simple.
\end{enumerate}
\end{theorem}

\begin{corollary}\label{cgenhl}
Let Assumption \ref{agenhl} hold. Then $\bar f$ and $\bar f^*$ are the unique eigenfunctions of $T$ and $T^*$ belonging to $E_+$ and $E_+^*$.
\end{corollary}

Assumptions \ref{agenhl}(ii) and \ref{agenhl}(iii) are very high-level and may be difficult to motivate in an economic context. However it is possible to provide more tractable sufficient conditions for these assumptions. These sufficient conditions also yield a stronger result than Theorem \ref{agenhl} and Corollary \ref{cgenhl} with regards to the separation of $r(T)$ in the spectrum of $T$, which is useful for nonparametric estimation of the eigenfunction and its eigenvalue. 

\begin{assumption}\label{agen} $T : E \to E$ is a bounded linear operator such that:\\[-20pt]
\begin{enumerate}
\item $T$ is positive
\item for each $f \in (E_+ \setminus \{0\})$ there exists $n = n(f) \geq 1$ such that $T^n f \in E_{++}$
\item $T^n$ is compact for some $n \geq 1$, and
\item $r(T) > 0$. 
\end{enumerate}
\end{assumption}

\begin{theorem} \label{tgen}
Let Assumption \ref{agen} hold. Then:\\[-20pt]
\begin{enumerate}
\item there exist $\bar f \in E_{++}$ and $\bar f^* \in E_{++}^*$ such that $T \bar f = r(T) \bar f$ and $T^* \bar f^* = r(T) \bar f^*$
\item $\bar f$ and $\bar f^*$ are the unique eigenfunctions of $T$ and $T^*$ belonging to $E_+$ and $E_+^*$
\item $r(T)$ is simple, isolated, and is the unique eigenvalue of $T$ on the circle $\{z \in \mathbb C : |z| = r(T)\}$.
\end{enumerate}
\end{theorem}

\section{Conclusion}

This paper provided conditions for nonparametric identification of positive eigenfunctions of operators in economics. The general theorems were applied to provide new identification conditions for external habit formation models and for positive eigenfunctions of pricing operators in dynamic asset pricing models. 

The identification conditions presented in this paper impose no smoothness or other shape restrictions on the eigenfunction beyond positivity. It remains to be seen to what extent the conditions may be relaxed when smoothness or additional shape restrictions are imposed. Further, the methods used in this paper might be extended to study nonparametric identification in other convex cones of functions, such as monotone, concave, or convex functions, which are of interest in economics.

\appendix
\setcounter{equation}{0}
\renewcommand{\theequation}{\Alph{section}.\arabic{equation}}

\section{Banach lattices and positive cones}  \label{blapp}

The following definitions are from \cite{Schaefer1999}. A vector space $E$ over the real field $\mathbb R$ endowed with an order relation $\leq$ is an ordered vector space if $f \leq g$ implies $f + h \leq g + h$ for all $f,g,h \in E$ and $f \leq g$ implies $\lambda f \leq \lambda g$ for all $f,g \in E$ and $\lambda \geq 0$. If, in addition, $\sup\{f,g\} \in E$ and $\inf\{f,g\} \in E$ for all $f,g \in E$ then $E$ is a vector lattice. If there exists a norm $\|\cdot\|$ on an ordered vector lattice $E$ under which $E$ is complete and $\|\cdot\|$ satisfies the lattice norm property, namely $f,g \in E$ and $|f| \leq |g|$ implies $\|f\| \leq \|g\|$, then $E$ is a Banach lattice. 

Let $E_+$ denote the positive cone of $E$ defined with respect to the order relation on $E$. If $E = L^p(\mathcal X, \mathscr X,\mu)$ with $1 \leq p \leq \infty$ then $E_+ = \{f \in E : f \geq 0$ a.e.-$[\mu]\}$. If $E = C_b(\mathcal X)$ or $C(\mathcal X)$ then $E_+ = \{f \in E : f(x) \geq 0$ for all $x \in \mathcal X\}$. An element $f \in E_+$ belongs to the quasi-interior $E_{++}$ of $E_+$ if $\{g \in E: 0 \leq g \leq f\}$ is a total subset of $E$. For example, if $E = L^p(\mathcal X, \mathscr X,\mu)$ with $1 \leq p < \infty$ and $f > 0$ a.e.-$[\mu]$ then $f \in E_{++}$. If $E = L^\infty(\mathcal X, \mathscr X,\mu)$ and $\mathrm{ess}\inf f > 0$ then $f \in E_{++}$. If $E = C_b(\mathcal X)$ or $C(\mathcal X)$ and $\inf_{x \in \mathcal X} f(x) > 0$ then $f \in E_{++}$.

Let $E^*$ denote the dual space of $E$ (the set of all bounded linear functionals on $E$). Given $f \in E$ and $f^* \in E^*$, let $\langle f,f^* \rangle$ denote the evaluation $f^*(f)$. The dual cone $E_+^*:= \{ f^* \in E^* : \langle f,f^* \rangle \geq 0$ whenever $f \in E_+ \}$ is the set of positive linear functionals on $E$. An element $f^* \in E^*_+$ is strictly positive if $f \in E_+$ and $f \neq 0$ implies $\langle f,f^* \rangle > 0$. The set of all strictly positive elements of $E_+^*$ is denoted $E_{++}^*$.

\section{Proofs } 

\begin{proof}[Proof of Theorem \ref{tL2}]
Immediate by application of Theorem \ref{tgen}.
\end{proof}

\begin{proof}[Proof of Corollary \ref{cL2}]
Let $\lambda \in \mathbb C$ denote the eigenvalue associated with $\bar f$. Assumption \ref{aL2}(ii) implies that $T^n \bar f = \lambda^n \bar f > 0$ a.e.-$[\mu]$ for some $n \geq 1$. Therefore, $|\lambda| > 0$ and so $r(T)> 0$, verifying Assumption \ref{aL2}(iv). The result follows by Theorem \ref{tL2}.
\end{proof}

\begin{proof}[Proof of Theorem \ref{tLinfty}]
Immediate by application of Theorem \ref{tgen}.
\end{proof}

\begin{proof}[Proof of Corollary \ref{cLinfty}]
Analogous to the proof of Corollary \ref{cL2}.
\end{proof}

\begin{proof}[Proof of Lemma \ref{suff conds}]
Take any $\psi \in \Psi$. Let $N_\psi = \{x \in \mathcal X : \psi(x) > 0\}$ denote the support of $\psi$. Assumption \ref{irred} shows that $\Pr(X_{t+n} \in N_\psi |X_t = x) > 0$ a.e.-$[Q]$ for some $n \geq 1$, and so $\mathbb E[\psi(X_{t+n})|X_t = x] > 0$ a.e.-$[Q]$. Assumption \ref{noarb} then implies that $T_n \psi(x) > 0$ a.e.-$[Q]$.
\end{proof}

\begin{proof}[Proof of Remark \ref{esp}]
To simplify notation, let $m_n =\prod_{i=0}^{n-1} m(X_{t+i},X_{t+1+i},Y_{t+1+i})$. Suppose the assertion is false. Then there exists $(\psi,n) \in \Psi \times \mathbb N$ such that $\mathbb E[\psi(X_{t+n})|X_t = x] > 0$ a.e.-$[Q]$ but for which $\mathbb E[m_n \psi(X_{t+n})|X_t = x] \leq 0$ for all $x \in \mathcal N$ where $\mathcal N \in \mathscr X$ and $Q(\mathcal N) > 0$. Let $1_{\mathcal N}(x)$ be the indicator function of the set $\mathcal N$. Then, by definition of $\mathcal N$ and the law of iterated expectations:
\begin{equation}\label{cont}
 0  \geq \mathbb E[1_{\mathcal N}(X_t)\mathbb E[m_n \psi(X_{t+n}) |X_t] ]  =\mathbb E[m_n 1_{\mathcal N}(X_t) \psi(X_{t+n})]\,.
\end{equation}
But $m_n > 0$ a.e.-$[\overline Q_n]$ by hypothesis. Moreover, and $1_{\mathcal N}(X_t) \psi(X_{t+n})$ is non-negative a.e.-$[\overline Q_n]$ and is positive on a set of positive $\overline Q_n$ measure (because $\mathbb E[\psi(X_{t+n})|X_t=x] > 0$ a.e.-$[Q]$ and $\psi \in \Psi$). Therefore, $m_n 1_{\mathcal N}(X_t) \psi(X_{t+n})$ is non-negative a.e.-$[\overline Q_n]$ and is positive on a set of positive $\overline Q_n$ measure. This implies $\mathbb E[m_n 1_{\mathcal N}(X_t) \psi(X_{t+n})] > 0$, which contradicts (\ref{cont}). Therefore, $\mathbb E[m_n \psi(X_{t+n})|X_t = x] > 0$ a.e.-$[Q]$, as required.
\end{proof}

\begin{proof}[Proof of Theorem \ref{id na disc}]
The result is proved by applying of Theorem \ref{tL2} to $T$. Assumptions \ref{lrr cpt}, \ref{irred} and \ref{noarb} satisfy Assumptions \ref{aL2}(i)(ii)(iii). Moreover, Assumption \ref{aL2}(iv) is satisfied because $r(T) \geq \rho$ (by definition of $r(T)$) and $\rho > 0$ because $T_n \phi = \rho^n \phi > 0$ a.e.-$[Q]$  for some $n \in \mathbb N$ by Lemma \ref{suff conds}. 
\end{proof}

\begin{proof}[Proof of Lemma \ref{nondeg proof}]
Let $1(\cdot) : \mathcal X \to \mathbb R$ denote the constant function, i.e. $1(x) = 1$ for all $x \in \mathcal X$. Under Assumption \ref{nondeg}, 
\begin{equation}
 T^n 1(x) = T_n 1(x) = (1+y_n(x))^{-n} \geq (1+C)^{-n} \times 1(x) 
\end{equation}
a.e.-$[Q]$. The result follows by applying Proposition 3 of \cite{Schaefer1960}.
\end{proof}

\begin{proof}[Proof of Theorem \ref{ex id na disc}]
Immediate by application of Theorem \ref{tL2}.
\end{proof}

\begin{proof}[Proof of Lemma \ref{noncompact}]
Let $\|\cdot\|$ denote the $L^2(\mathcal X,\mathscr X,Q)$ norm and let $\{\psi_i\}_{i \in \mathbb N}$ be a sequence of functions belonging to the ball $\{\psi \in L^2(\mathcal X,\mathscr X,Q) : \psi(x_t) = \psi(w_{t-1},\ldots,w_{t-\ell+1}), \|\psi\| = 1\}$ such that $\{\psi_i\}_{i \in \mathbb N}$ has no convergent subsequence (we can always choose such a sequence because this ball is not compact). Observe that $T \psi_i (x_t) = \psi(w_t,\ldots,w_{t-\ell+2})\mathbb E[m(X_t,X_{t+1},Y_{t+1})|X_t = x_t]$ for each $\psi_i$. The proof proceeds by contradiction: assume $T$ is compact. Then $\{T \psi_i\}_{i \in \mathbb N}$ has a convergent subsequence, say $\{T \psi_{i_k}\}_{k \in \mathbb N}$. Moreover,
\begin{eqnarray}
 \!\!\!\! & & \|T\psi_{i_k} - T \psi_{i_j}\|^2 \notag \\
 \!\!\!\! & &\quad = \quad \mathbb E\left\{ \left[ (\psi_{i_k}(W_t,\ldots,W_{t-\ell+2})-\psi_{i_j}(W_t,\ldots,W_{t-\ell+2}))\mathbb E[m(X_t,X_{t+1},Y_{t+1})|X_t] \right]^2 \right\} \\
 \!\!\!\! & & \quad \geq \quad c^2 \|\psi_{i_k} - \psi_{i_j}\|^2.
\end{eqnarray}
The subsequence $\{T \psi_{i_k}\}_{k \in \mathbb N}$ is convergent and therefore Cauchy, and so the subsequence $\{\psi_{i_k}\}_{k \in \mathbb N}$ is Cauchy by the above inequality. Therefore, $\{\psi_{i_k}\}_{k \in \mathbb N}$ must be convergent (by completeness of $L^2(\mathcal X,\mathscr X,Q)$). This contradicts the fact that $\{\psi_i\}_{i \in \mathbb N}$ has no convergent subsequence.
\end{proof}

\begin{proof}[Proof of Lemma \ref{lem-irred}]
Let $\mathcal X = \mathbb R^\ell$, $\mathscr X$ be the Borel $\sigma$-algebra on $\mathbb R^\ell$, and $Q$ be the stationary distribution of $\{X_t\}_{t=-\infty}^\infty$ on $\mathbb R^\ell$ (which exists by stationarity of $\{W_t\}_{t=-\infty}^\infty$). Let $S \in \mathscr X$ have positive Lebesgue measure (and therefore non-negative $Q$-measure). \cite{ChanTong} show that positivity of the density $f_U$ of $u$ implies $\mathbb P(X_{t+\ell} \in S | X_t = x) > 0$ for all $x \in \mathbb R^\ell$. Therefore, Assumption \ref{irred} is satisfied with $n = \ell$. 
\end{proof}

\begin{proof}[Proof of Theorem \ref{mktthmc}]
Apply Theorem \ref{id na disc} with $T_{\bar \tau}$ in place of $T$.
\end{proof}

\begin{proof}[Proof of Theorem \ref{t-habit}]
Immediate by application of Theorem \ref{tW} and Remark \ref{esp}.
\end{proof}

\begin{proof}[Proof of Corollary \ref{cgenhl}]
The eigenvalue $r(T)$ is simple by Theorem \ref{tgenhl}. Therefore any linearly independent eigenfunction must correspond to an eigenvalue different from $r(T)$. Suppose $\bar g \in (E_+ \setminus \{0\})$ is an eigenfunction of $T$ with eigenvalue $\lambda \neq r(T)$. Then, for $f^* \in E^*$, $f \in E$,
\begin{equation}
 \lambda \langle \bar g,\bar f^* \rangle = \langle T \bar g, \bar f^* \rangle = (\bar f^*\circ T)(\bar g) = T^*(\bar f^*)(\bar g) = r(T)\bar f^*(\bar g) = r(T) \langle \bar g,\bar f^* \rangle
\end{equation}
which is a contradiction, since $\bar f^* \in E^*_{++}$ and $g \in (E_+ \setminus \{0\})$ implies $\langle \bar g,\bar f^* \rangle > 0$. 

As the multiplicities of eigenvalues of finite multiplicity of $T$ are preserved under adjoints \cite[Exercise VII.5.35]{DunfordSchwartz}, $r(T)$ is also a simple eigenvalue of $T^*$. That $\bar f^* \in E_{++}^*$ is the unique eigenfunction of $T^*$ in $E_+^*$ follows by a similar argument.
\end{proof}

\begin{proof}[Proof of Theorem \ref{tgen}]
Theorem \ref{agenhl} and Corollary \ref{cgenhl} will be used to prove parts (i) and (ii). It is enough to show that $T$ satisfies the Assumptions of Theorem \ref{agenhl} and Corollary \ref{cgenhl}. Assumption \ref{agenhl}(ii) is satisfied under Assumptions \ref{agen}(i)(ii)(iv) because $r(T)\in \sigma(T)$ since $T$ is positive \cite[p. 312]{Schaefer1999} and $\sigma(T) \ni r(T) > 0$ and $T^n$ compact together imply $r(T)$ is a pole of the resolvent of $T$ \cite[Theorem 6, p. 579]{DunfordSchwartz}. Assumption \ref{agen}(ii) implies that for each $f \in (E_+ \setminus \{0\})$ and $f^* \in E_+^*$ there exists $n \geq 1$ such that $\langle T^n f,f^* \rangle > 0$, so $T$ is irreducible \cite[Proposition III.8.3]{Schaefer1974}. This proves parts (i) and (ii).

For part (iii), compactness of $T^n$ and the fact that $\sigma(T)^n = \sigma(T^n)$ implies that the only possible limit point of elements of $\sigma(T)$ is zero \cite[p. 579]{DunfordSchwartz}. Therefore $r(T)$ is an isolated point of $\sigma(T)$, and theorem \ref{agenhl} implies $r(T)$ is simple. Finally, that $r(T)$ is the unique eigenvalue belonging to the circle $\{z \in \mathbb C : |z| = r(T)\}$ follows under Assumption \ref{agen}(ii) by Proposition V.5.6 of \cite{Schaefer1974}.
\end{proof}

{
\begin{spacing}{1.0}
\bibliographystyle{chicago}
\bibliography{refsmain-id}
\end{spacing}
}

\end{document}